\documentclass[lettersize,journal]{IEEEtran}

\usepackage{amsmath,amsfonts,amssymb,bbm, amsthm, xfrac}
\usepackage{algorithmic}
\usepackage{algorithm}
\usepackage{array, multirow}
\usepackage{caption, subcaption}
\usepackage{textcomp}
\usepackage{stfloats}
\usepackage{url}
\usepackage{verbatim}
\usepackage{graphicx}
\usepackage{cite}
\usepackage{caption}
\usepackage{subcaption}
\theoremstyle{plain}
\newtheorem{theorem}{Theorem}

\usepackage{color, soul}

\begin{document}

\title{Distributed Gaussian Mixture PHD Filtering under Communication Constraints}

\author{Shiraz Khan, Yi-Chieh Sun, and Inseok Hwang\thanks{The authors are with the School of Aeronautics and Astronautics, Purdue University, West Lafayette, IN 47906, USA. (email: \texttt{shiraz@purdue.edu})}
\thanks{\textcolor{blue}{© 2023 IEEE. Personal use of this material is permitted. Permission from IEEE must be obtained for all other uses, in any current or future media, including reprinting/republishing this material for advertising or promotional purposes, creating new collective works, for resale or redistribution to servers or lists, or reuse of any copyrighted component of this work in other works.}}
}

\markboth{IEEE Transactions on Signal Processing}%
{Khan \MakeLowercase{\textit{et al.}}: Distributed Gaussian Mixture PHD Filtering under Communication Constraints}


\maketitle

\begin{abstract}
The Gaussian Mixture Probability Hypothesis Density (GM-PHD) filter is an almost exact closed-form approximation to the Bayes-optimal multi-target tracking algorithm. Due to its optimality guarantees and ease of implementation, it has been studied extensively in the literature. However, the challenges involved in implementing the GM-PHD filter efficiently in a distributed (multi-sensor) setting have received little attention. The existing solutions for distributed PHD filtering either have a high computational and communication cost, making them infeasible for resource-constrained applications, or are unable to guarantee the asymptotic convergence of the distributed PHD algorithm to an optimal solution.     
In this paper, we develop a distributed GM-PHD filtering recursion that uses a probabilistic communication rule to limit the communication bandwidth of the algorithm, while ensuring asymptotic optimality of the algorithm. We derive the convergence properties of this recursion, which uses weighted average consensus of Gaussian mixtures (GMs) to lower (and asymptotically minimize) the Cauchy-Schwarz divergence between the sensors' local estimates. In addition, the proposed method is able to avoid the issue of false positives, which has previously been noted to impact the filtering performance of distributed multi-target tracking. Through numerical simulations, it is demonstrated that our proposed method is an effective solution for distributed multi-target tracking in resource-constrained sensor networks.
\end{abstract}


\maketitle

\section{Introduction}
\label{sec:introduction}
\IEEEPARstart{M}{ulti-target} tracking refers to the problem of estimating the states of an unknown, time-varying number of dynamical systems (which are called \textit{targets}) using noisy measurements.
Multi-target tracking algorithms have received considerable attention in the last two decades, with new domains of application being formed, such as tracking of multiple autonomous vehicles or monitoring pedestrian traffic using cameras \cite{rangesh2019no, fu2019multi}. In particular, the Gaussian Mixture Probability Hypothesis Density (GM-PHD) filter is popular for its ease of analysis and implementation \cite{vo2006gmphd}. It models the uncertain multi-target state as a Poisson random finite set (RFS), whose intensity is represented using a Gaussian mixture (GM), and is able to track the targets without requiring explicit measurement-to-track or track-to-track data associations.
As the GM-PHD filter was developed on a solid theoretical foundation, its optimality and convergence guarantees are better understood than heuristic algorithms or data-driven approaches to multi-target tracking. 

In the multi-sensor multi-target tracking problem, a network of spatially separated sensors are required to carry out multi-target tracking cooperatively. By sharing information across the communication channels of the sensor network, the uncertainty in the multi-target estimate of each sensor can be reduced \cite{gostar2017cauchy}. This is especially true when the sensors have different sensing capabilities or areas of coverage, such as cameras mounted on drones which are overlooking different parts of the surveillance region \cite{battistelli2020differentFOVs}. A naive extension of the GM-PHD filter to the multi-sensor setting requires an excessive amount of communication to be carried out between the sensors. A much more efficient solution involves running local GM-PHD filters in parallel, followed by a fusion step for combining the sensors' estimates \cite{li2018cardinality, li2018partial, wu2022partial}. This approach is analogous to that of multi-sensor Kalman filtering with consensus or diffusion steps for single-target tracking \cite{talebi2018distributed}, but has the additional capability of being able to track a dynamically varying number of targets without requiring explicit measurement-to-track or track-to-track data association \cite{li2018partial}.

Recently, several authors have observed that the multi-sensor multi-target tracking problem can be solved by combining local GM-PHD filtering with a weighted arithmetic average (WAA) fusion step, in which each sensor updates its estimate to a weighted average of the estimates of the sensor network. WAA fusion can be realized in a distributed manner by using one or more weighted average consensus iterations, such that each consensus iteration only requires the sensors that are connected to communicate with each other.
In \cite{li2018cardinality}, the authors used an average consensus step to fuse the estimated number of targets in the surveillance region (i.e., the estimated cardinality of the RFS), which was shown to improve the overall local PHD filtering performance. In \cite{li2020parallel, li2021is}, and \cite{wu2022partial}, WAA fusion was used to fuse the posterior intensities of the local GM-PHD filters, rather than just the cardinality estimates. The theoretical justification for WAA fusion of intensities was noted by the authors of \cite{battistelli2020differentFOVs}, who observed that it minimizes a cost function based on the Cauchy-Schwarz divergence between RFSs. The authors of \cite{gao2020multiobject} and \cite{gostar2017cauchy} have further connected WAA fusion to other types of cost functions, while contrasting it with the alternative fusion strategy: weighted geometric average (WGA) fusion. Notably, WAA fusion is easier to implement and leads to less false negatives (i.e., missed detections) than WGA fusion, resulting in more consistent cardinality estimates \cite{li2021cphd, gostar2020cooperative}. In \cite{julier2012conservative}, it was noted that WAA fusion also has a desirable consistency/conservativeness property which makes it robust in the presence of unknown cross-correlations between the sensors in a multi-sensor network. However, to the best of our knowledge, none of the existing works has discussed the convergence properties of the weighted average consensus approach, in which consensus is sought over real-valued functions; this is different from consensus of vectors in Euclidean spaces, which has been well-studied in the literature \cite{olfati2007consensus,corless2012consensus}.

The weighted average consensus approach requires each sensor to broadcast all of its GM components to its neighbors, which typically outnumber the targets in the surveillance region, leading to the algorithm incurring a high communication cost at the sensors. Moreover, WAA fusion is known to suffer from a high number of false positives (i.e., the misidentification of background clutter as targets), which can be especially problematic in distributed sensor networks, as the false positives can get amplified due to the double-counting of information between the sensors of the network \cite{li2020arithmetic}. 
To address both of these issues simultaneously, the authors of \cite{li2018partial, li2020parallel, li2021is}, and \cite{wu2022partial} proposed the use of \textit{partial consensus}, in which the sensors do not communicate weak GM components that likely correspond to the false positives. This is achieved either by choosing a threshold at each sensor, such that GM components weaker than this threshold are not used in the inter-sensor communication, or by using only the highest-weighted GM components in the consensus step. 
An apparent drawback of this approach is that the partial consensus step never converges to the weighted arithmetic average, as the GM components smaller than the threshold are never communicated between sensors.

In this paper, we develop a distributed GM-PHD algorithm for multi-sensor multi-target tracking which addresses each of the above limitations of the existing algorithms. To limit the communication bandwidth requirement of the proposed distributed GM-PHD algorithm, we develop a random sampling method for selecting the GM components that are transmitted by each sensor during the average consensus step. The sampling probabilities are designed based on two considerations -- similar to the existing distributed GM-PHD algorithms, the sensors avoid communicating small GM components which are likely to correspond to the false positives, thereby reducing the likelihood of false positives. However, unlike the existing algorithms, the proposed algorithm ensures that the sensors' estimates asymptotically converge (in expectation) to the WAA, which is the optimal multi-sensor multi-target estimate in terms of the Cauchy-Schwarz divergence.
The convergence properties of the proposed algorithm are established mathematically rigorously by conducting an error analysis of consensus of functions in $L^p$ spaces, which has not been studied previously in the literature.
By numerically comparing the proposed distributed GM-PHD algorithm to existing algorithms, it is demonstrated that the proposed algorithm strikes a desirable trade-off between the optimality and the communication bandwidth requirement of multi-sensor multi-target tracking, making it especially well-suited for use in resource-constrained sensor networks.

The rest of the paper is organized as follows. In Section \ref{sec:prob_form}, the multi-sensor multi-target tracking problem is formulated, and the local GM-PHD filtering recursion is described. Section \ref{sec:fusion_step} develops the proposed inter-sensor fusion approach for distributed GM-PHD filtering, wherein we conduct an error analysis of the weighted average consensus step. The proposed random sampling rule is developed in Section \ref{sec:comm_constr}. In Section \ref{sec:simulation}, we compare the proposed algorithm to some of the existing multi-sensor GM-PHD filtering algorithms in the literature. Finally, Section \ref{sec:conclusion} presents the conclusion.

\section{Problem Formulation}
\label{sec:prob_form}
\subsection{Target Model}
In the multi-sensor multi-target tracking problem, an unknown time-varying number of dynamical systems, called the \textit{targets}, are observed using a sensor network.
The state of the $l^{th}$ target at timestep $k$ is given by a random vector $x_{k}^{(l)} \in \mathcal{X}$, where $\mathcal{X} = \mathbb{R}^{d_x}$ is the state space of each target and $d_x$ is its dimension. The collection of the target states at timestep $k$ is modeled as a random finite set (RFS)\footnote{A brief introduction to RFSs and their properties can be found in \cite{vo2006gmphd}.} $X_k = \left\{x_{k}^{(1)},...,x_{k}^{(|X_k|)}\right\}$, where $|{}\cdot{}|$ denotes the cardinality of a set.
Thus, the number of targets, $|X_k|$, is an integer-valued random variable.

The state of each target is assumed to evolve independently of the other targets, according to a linear Gauss-Markov process on the state space $\mathcal{X}$ with the transition kernel $f_{k|k-1}$. Thus, given that a target exists at timesteps $k-1$ and $k$, the probability density function (pdf) of its state at timestep $k$ is
\begin{equation}
    f_{k|k-1}(x | x_{k-1}) = 
    \mathcal{N}(x; F_{k-1} x_{k-1}, Q_{k-1})
    \label{eq:dynamic_model}
\end{equation}
where $F_{k-1}$ and $Q_{k-1}$ are known matrices and $x_{k-1} \in\mathcal X$ is the state of the target at timestep $k-1$. 
$\mathcal{N}({}\cdot{}; m, P)$ denotes the pdf of a multivariate Gaussian random vector with mean $m\in \mathbb R^{d_x}$ and covariance $P\in \mathbb R^{d_x\times d_x}$. 
The probability of a target continues to exist across both timesteps is called its \textit{survival probability}, denoted by $p^S_{k}:\mathcal X \rightarrow [0, 1]$, and is assumed to depend only on the current state of each target.

New targets of interests can arise due to spontaneous \textit{births} or \textit{spawns}. Births are independent of $X_k$ and assumed to be sampled from an underlying Poisson RFS having the intensity $\gamma _k:\mathcal X \rightarrow \mathbb R_{\geq0}$, where $\mathbb R_{\geq0}$ denotes the set of non-negative real numbers. 
Spawns correspond to targets that emerge from the vicinity of existing targets, so they are dependent on $X_k$. Given that a target exists at $x_{k-1} \in \mathcal X$ at timestep $k-1$, the new targets spawned from it are assumed to be sampled from
a Poisson RFS having the intensity $\sigma _{k|k-1}({}\cdot{}|x_{k-1}):\mathcal X \rightarrow \mathbb R_{\geq 0}$. Finally, note that the transition kernel $f_{k|k-1}$ and spawn intensity $\sigma_{k|k-1}$ are independent of the labels or identities of the targets; if they are dependent on the target labels, a labelled RFS can be used to represent the multi-target state instead \cite{vo2014labeled}.

\subsection{Sensor Network Model}
The targets are being observed by
a sensor network consisting of a number of spatially distributed sensors, which can be represented as a directed graph $\mathcal G=(\mathcal V, \mathcal E)$, where $\mathcal V$ denotes its vertices (representing the sensors) and $\mathcal E \subseteq \mathcal V \times \mathcal V$ is the set of edges (representing the directional communication links between sensors).
The graph $\mathcal G$ is assumed to be strongly connected. 
For each sensor $i \in \mathcal{V}$, ${N}^-_{i} = \left\{ j\in\mathcal{V}: (j,i)\in\mathcal{E} \right\}$ denotes its set of in-neighbors, i.e., the set of sensors from which sensor $i$ can receive data. Similarly, $N^+_i$ is the set of out-neighbors.

At timestep $k$, each sensor $i$ obtains a set of measurements $Z_{i,k} = {\left\{z^{(1)}_{i,k},...,z_{i,k}^{(|Z_{i,k}|)}\right\}}$, where $z_{i,k}^{(l)} \in \mathcal{Z}$. Here, $\mathcal{Z} \subseteq \mathbb{R}^{d_z}$ is the measurement space and $d_z$ is its dimension. It is assumed that the measurements are generated as per the following mechanism: the $i^{th}$ sensor obtains a measurement from a target at $x \in \mathcal X$ with the probability $p^D_{i,k}(x)$, where $p^D_{i,k}:\mathcal X \rightarrow [0,1]$ is called the detection probability of sensor $i$ at timestep $k$. The event of the $l^{th}$ target being detected at the $i^{th}$ sensor is independent of the detections at other target-sensor pairs.
Given that a target at $x \in \mathcal X$ is detected by sensor $i$, the pdf of the obtained measurement is denoted as $g_{i,k}({}\cdot{}|x)$ and is  
given by the linear Gaussian model
\begin{equation}
    g_{i,k}(z|x) = \mathcal N(z; H_{i,k} x, R_{i,k})
    \label{eq:measurement_model}
\end{equation}
where $H_{i,k}$ and $R_{i,k}$ are known matrices.
In addition, the obtained measurement can correspond to clutter, which are not generated by any of the targets, but arise due to sensor noise or other random effects. The clutter measurements are assumed to be realizations of a Poisson RFS with the intensity $\kappa_i:\mathcal Z \rightarrow \mathbb R_{\geq 0}$.
Thus, $\lbrace Z_{i,k} \rbrace$ are realizations of an RFS, and the distribution of the cardinality of this RFS (i.e., the number of measurements) is determined by the detection probability $p^D_{i,k}$ and clutter intensity $\kappa _i$.

\subsection{Local GM-PHD Filtering}
\label{sec:subsec_local_filt}
Distributed multi-target tracking over a sensor network is usually accomplished using a succession of a local filtering step (e.g., using a PHD filter) at each of the sensors followed by a fusion step, wherein the sensors share information across the communication channels. 
The predicted multi-target state of sensor $i$ is an RFS whose intensity is denoted as $v_{i,k|k-1}$. It characterizes the information known about the multi-target state $X_k$ at sensor $i$ at timestep $k$, before observing the measurements $Z_{i,k}$. After the incorporation of the measurements $Z_{i,k}$, the updated (posterior) intensity is denoted as $v_{i,k|k}$. 

In the proposed distributed GM-PHD filter, each sensor performs the local prediction and measurement update steps according to the GM-PHD filtering algorithm \cite{vo2006gmphd}. The GM-PHD filter uses Gaussian mixtures (GMs) to represent the intensity functions, admitting a closed-form solution to the PHD recursion. Consequently, the GM-PHD filter is optimal in a Bayesian sense, under reasonable assumptions on the target and sensor models; a more complete discussion of the relevant assumptions may be found in \cite{vo2006gmphd}. In practice, the GM-PHD filter is approximated by pruning small GM components to keep the space complexity of the algorithm from growing over time.

Given $x_{k-1} \in \mathcal X$, let the prior intensity $v_{i,k|k-1}$
be represented using a GM with $J_{i,k|k-1}$ components, as follows:
\begin{align}
    v_{i,k|k-1}(x) &= \sum_{l=1}^{J_{i,k|k-1}}w^{(l)}_{i,k|k-1}\mathcal N(x;m^{(l)}_{i,k|k-1}, P^{(l)}_{i,k|k-1})
    \label{eq:prior_intensity}
\end{align}
Sensor $i$'s prior estimate of the number of targets is $\sum_{l=1}^{J_{i,k|k-1}} w_{i,k|k-1}^{(l)}$, as this sum corresponds to the integral of the intensity $v_{i,k|k-1}(x)$.

The birth and spawn intensities ($\gamma_k$ and $\sigma_{k|k-1}({}\cdot{}|x_{k-1})$, respectively) are also represented by GMs whose parameters are assumed to be known. The prior intensity $v_{i,k|k-1}$ includes the surviving targets from timestep $k-1$ as well as new targets that were born or spawned at timestep $k$. Since these correspond to independent Poisson RFSs, we can add the corresponding intensity functions as follows:
\begin{equation}
    v_{i,k|k-1}(x) = v_{i,k|k-1}^S(x) + v_{i,k|k-1}^\sigma(x) + \gamma_{k}(x),
\end{equation}
where $v^S_{i,k|k-1}$ corresponds to the targets that survived from timesteps $k-1$ to $k$ and $v_{i,k|k-1}^{\sigma}$ corresponds to newly spawned targets. Similarly, the posterior intensity is computed as a sum of two intensities:
\begin{equation}
    v_{i,k|k}(x) = (1 - p_{D,k}(x)) v_{i,k|k-1}(x) + \sum_{z \in Z_{k}} v_{i,k}^D(x;z)
\end{equation}
where the first term corresponds to the targets which were not detected and $v^{D}_{i,k}(x;z)$ is the new information obtained through the measurement $z\in Z_{k}$. Further details about the computation of prior and posterior intensities ($v_{i,k|k-1}$ and $v_{i,k|k}$, respectively) can be found in \cite{vo2006gmphd}. 

We refer to the computation of $v_{i,k|k-1}$ and $v_{i,k|k}$ using the local measurements $Z_{i,k}$ at sensor $i$ as local GM-PHD filtering.
In the remainder of this paper, we focus on developing a distributed GM dissemination and fusion protocol for improving the multi-target tracking performance of each sensor. Our proposed distributed GM-PHD filtering algorithm constitutes the local GM-PHD filtering step followed by one or more inter-sensor fusion steps. 

\section{Fusion of Local GM-PHD Filters}
\label{sec:fusion_step}

In addition to local GM-PHD filtering, the communication channels, i.e., the edges in $\mathcal E$, can serve as additional sources of information for each sensor.
In theory, the optimal solution to the given multi-sensor multi-target problem can be realized by aggregating all the measurements of the sensor network as $\bigcup_{i\in \mathcal V}Z_{i,k}$ at each sensor, at each timestep. Since the communication and computational cost of such an approach does not scale well with the number of sensors in the network (which is equal to $|\mathcal V|$), it is infeasible when one or more of the following challenges are taken into consideration: 
\begin{enumerate}
    \item the total number of sensors in the network is large,
    \item the inter-sensor communication channels have limited bandwidths, or 
    \item the communication is not instantaneous, i.e., there are significant delays between the transmission and reception of inter-sensor communications.
\end{enumerate}
A distributed solution to the multi-sensor multi-target tracking problem can solve each of the preceding challenges. Distributed PHD filtering is achieved by fusing the intensity functions $v_{i,k|k}$ of the individual sensors, as these are typically much smaller in dimension than the measurement sets $Z_{i,k}$. For instance, the measurements in $Z_{i,k}$ may be high-dimensional camera images, whereas $v_{i,k|k}$ has a concise representation in terms of its GM components.

\subsection{Weighted Arithmetic Average (WAA) Fusion}
The weighted arithmetic average (WAA) of the local posterior intensity functions of the sensors is defined as
\begin{equation}
    v^{*}_{k|k}(x) = \sum_{i \in \mathcal V} \omega_i v_{i,k|k}(x)
    \label{eq:WAA_def}
\end{equation}
where $\omega_i \geq 0$ and $\sum_{i\in \mathcal V} \omega_i = 1$. In \cite{da2020kullback}, it was shown that the WAA (\ref{eq:WAA_def}) minimizes the weighted Kullback-Leibler (KL) divergence between itself and the local posterior intensities, $v_{i,k|k}$, when the fused intensity is taken as the second argument of the divergence. Similarly, it was shown in \cite{gostar2017cauchy} that the WAA also minimizes the weighted Cauchy-Schwarz (CS) divergence (which is symmetric in its arguments). Thus, we have
\begin{align}
    v^{*}_{k|k}(x) &= \arg \min_{g} \left( \sum_{i\in \mathcal V} \omega_i D_{KL}\left(v_{i,k|k}(x)\mathrel{\Vert}g(x)\right)\right) \\&= \arg \min_{g}\left( \sum_{i\in \mathcal V} \omega_i D_{CS}\left(v_{i,k|k}(x)\mathrel{\Vert}g(x)\right)\right)
    \label{eq:csd}
\end{align}
where $D_{KL}$ and $D_{CS}$ denote the KL and CS divergences, respectively.
Either choice of divergence can be considered as a cost function for the local PHD fusion process, analogous to how the weighted Euclidean distance is used as the cost function in single-target tracking. By minimizing the divergence, the information gain between the local posterior intensities and the fused intensities is minimized; so the WAA (\ref{eq:WAA_def}) is optimal in terms of the \textit{principle of minimum discrimination of information (PMDI)} \cite{gostar2017cauchy, gao2020multiobject}.

The weights $\omega_i$ can be chosen as $\omega_i =\sfrac{1}{|\mathcal V|}$ if the sensor network is homogeneous, i.e., all the sensors have the same sensing capability, measurement noise covariance and detection probability. When this is not the case, a larger weight can be assigned to sensors having better sensing capability. The connectivity of the sensor network $\mathcal G$ can also inform the choice of $\omega _i$. 

Additionally, note that by taking an integral on both sides of (\ref{eq:WAA_def}), we have

\begin{equation}
    \int_{\mathcal X} v^{*}_{k|k}(x)\hspace{1pt}dx = \sum_{i \in \mathcal V} \omega_i \left(\int_{\mathcal X} v_{i,k|k}(x)\hspace{1pt}dx\right)
    \label{eq:integral_WAA}
\end{equation}
From the properties of Poisson RFS intensities, it follows that the integral of an intensity is the expected value of the cardinality of the corresponding RFS. Thus, WAA fusion of the intensities $v_{i,k|k}$ entails WAA fusion of the cardinality estimates of the sensors. The converse is not true, as two different functions may integrate to the same value. Thus, the proposed approach is different from the one used in \cite{li2018cardinality}, which directly computes the WAA of the cardinality estimates.

\subsection{Distributed WAA Fusion using Consensus}

To realize WAA fusion at each sensor in a distributed manner, an average (weighted) consensus algorithm can be used. In each iteration of the average consensus algorithm, the sensors communicate the GM components of their posterior intensity to their out-neighbors.
Thereafter, each sensor fuses its posterior intensities with those of its in-neighbors using a weighted combination, where the weights $\Omega_{ij}\in \mathbb R$ correspond to the entries of a matrix $\Omega = [\Omega_{ij}]$, with $\Omega_{ij}\neq 0$ if and only if $(j,i)\in \mathcal E$.
The average consensus algorithm is described in Algorithm \ref{alg:con}, which uses a total of $\alpha \geq 1$ inter-sensor fusion iterations.

\begin{algorithm}
\caption{Average Consensus of Posterior Intensities}
\begin{algorithmic}[1]
\vspace{2pt}
\REQUIRE The consensus weights $\Omega_{ij}$, posterior intensities $v_{i,k|k}$, and number of inter-sensor fusion steps $\alpha \geq 1$.\\
\vspace{3pt} 
\STATE Set $v_{i,k|k}^{(0)} \leftarrow v_{i,k|k} \forall i\in \mathcal V$
\FOR{$l=1,2,\dots,\alpha$}
\STATE Each sensor communicates its fused posterior intensity $v^{(l-1)}_{j,k|k}(x)$ to its out-neighbors, $N^+_i$. \vspace{2pt}
\STATE Update the posterior intensities as \begin{equation}
    v^{(l)}_{i,k|k}(x) \leftarrow \sum_{j \in N^{-}_i \cup \lbrace i\rbrace} \Omega_{ij} \hspace{2pt} v^{(l-1)}_{j,k|k}(x)
    \label{eq:average_consensus}
\end{equation}
$\forall i\in\mathcal V$.
\ENDFOR
\RETURN The fused posterior intensity, $v^{(\alpha)}_{i,k|k}(x)$.
\end{algorithmic}
\label{alg:con}
\end{algorithm}

\subsubsection{Error Analysis of Algorithm \ref{alg:con}}
\label{sec:subsec_error_analysis}
At the $l^{th}$ iteration of the average consensus algorithm (Algorithm \ref{alg:con}),
the error between the $i^{th}$ sensor's posterior intensity and the WAA is $v^{(l)}_{i,k|k}-v^*_{k|k}$. To establish the convergence properties of the algorithm, we need to show that the average consensus step (\ref{eq:average_consensus}) drives this error to $0$ (i.e., $v^{(l)}_{i,k|k}-v^*_{k|k}$ approaches the function which is $0$ everywhere in its domain) at each of the sensors.
To facilitate this analysis, let us assume that the intensities $v_{i,k|k}$ are elements in the $L^p(\mathcal X)$ function space, i.e., for some $1\leq p\leq \infty$,
\begin{equation}
\|v_{i,k|k}\|_{L^p} =  \left(\int_\mathcal X |v_{i,k|k}(x)|^p dx\right)^{\frac{1}{p}}< \infty
\end{equation}
$\forall i\in \mathcal V$. Recall that the integral of $v_{i,k|k}$ is sensor $i$'s posterior estimate of the cardinality of the RFS $X_k$. If $p=1$, then $v_{i,k|k}\in L^1(\mathcal X)$ if and only if sensor $i$'s cardinality estimate is bounded. Hence, the assumption that $v_{i,k|k}\in L^p(\mathcal X)$  is reasonable.

We define the vector space $V = \left(L^p(\mathcal X)\right)^{|\mathcal V|}$, so that we can collectively represent the set of intensities $\lbrace v_{i,k|k}\rbrace_{i\in\mathcal V}$ as a vector in $V$:
\[\bar v_{k|k} = \begin{bmatrix} v_{1, k|k} & v_{2, k|k} & \dots & v_{|\mathcal V|, k|k}\end{bmatrix}^\intercal \in V\]
For vectors in $V$, the addition and scalar multiplication operations are defined as the pointwise addition and pointwise scalar multiplication of their components.
We endow $V$ with a norm $\|\mathrel\cdot\|_V$, defined as
\begin{equation}
\| \bar v_{k|k}\|_V = \left(\sum_{i=1}^{|\mathcal V|} \|v_{i,k|k}\|_{L^p} ^p\right)^{\frac{1}{p}}
\end{equation}
Thus, $(V, \|\mathrel\cdot\|_V)$ is a normed vector space. For a matrix $A$, let $\|A\|_{\textsc{op}}$ refer to its operator norm\footnote{If we set $p=2$ in the foregoing definitions, then $V$ can be made into a Hilbert space, in which case $\|A\|_{\textsc{op}}$ is equal to the largest singular value of $A$ \cite[Thm. 4, p. 40]{halmos1957introduction}.}, given by
\begin{equation}
\|A\|_{\textsc{op}} = \sup \left\lbrace \frac{\|A v\|_V}{\|v\|_V} \mathrel{:} v\neq 0, v\in V \right\rbrace
\end{equation}
With the above definitions in place, we can state the conditions for the convergence of the average consensus step (\ref{eq:average_consensus}). 

\begin{theorem}
Suppose the following conditions are satisfied:
\begin{enumerate}
\item $\Omega \mathbbm 1 = \mathbbm 1$, where $\mathbbm 1 = \begin{bmatrix}
1 & 1 & \dots & 1\end{bmatrix}^\intercal$,
%
\item $\bar \omega^\intercal \Omega = \bar \omega^\intercal$, where $\bar \omega  = [\omega _1, \omega _2, \dots, \omega _{|\mathcal V|}]^\intercal$, and
\item $\|{\Omega - \mathbbm 1 \bar{\omega}^\intercal} \|_{\textsc{op}} < 1$,
\end{enumerate}
then 
the fused posterior intensities $\lbrace v_{i,k|k}^{(\alpha)} \rbrace_{i\in \mathcal V}$ computed using Algorithm \ref{alg:con}
asymptotically converge to the WAA in the $L^p$ norm:
\begin{align}
   \lim_{\alpha\rightarrow \infty} \| v^{(\alpha)}_{i,k|k} - v^{*}_{k|k}\|_{L^p} \rightarrow 0, \quad \forall i\in\mathcal V
    \label{eq:convergence}
\end{align}
As a consequence, there is a subsequence $\lbrace l_1, l_2, \dots \rbrace$ such that $\lim_{s\rightarrow \infty} v^{(l_s)}_{i,k|k} = v^{*}_{k|k}$ almost everywhere \cite[p. 75]{bartle2014elements}.
\end{theorem}
\begin{proof}
To show (\ref{eq:convergence}), we can rewrite the average consensus step (\ref{eq:average_consensus}) as $ \bar v^{(l)}_{k|k} = \Omega \bar v^{(l-1)}_{k|k}$. 
Observe that,  given $ l \geq 1$
\begin{equation}
\bar \omega ^\intercal \bar v^{(l)}_{k|k} = \bar \omega ^\intercal \Omega \bar v^{(l-1)}_{k|k} = \bar \omega ^\intercal \bar v^{(l-1)}_{k|k}
\end{equation}
By induction, we have that $\bar \omega ^\intercal \bar v^{(l)}_{k|k} = \bar \omega ^\intercal \bar v^{(0)}_{k|k} = v^*_{k|k}$; in other words, the weighted average is invariant under the average consensus step (\ref{eq:average_consensus}).
Thus, the error dynamics corresponding to the average consensus step is
\begin{align}
\bar v^{(l+1)}_{k|k} - \mathbbm 1 \bar v^*_{k|k} &= \Omega \bar v^{(l)}_{k|k} - \mathbbm 1  \bar \omega ^\intercal \bar v^{(l)}_{k|k}\\
&=\left(\Omega  - \mathbbm 1  \bar \omega ^\intercal \right) \bar v^{(l)}_{k|k} \\
&=\left(\Omega  - \mathbbm 1  \bar \omega ^\intercal \right) (\bar v^{(l)}_{k|k}-\mathbbm 1 \bar v^*_{k|k})
\end{align}
where, in the last step, we used the facts that $\Omega \mathbbm 1 = \mathbbm 1$ and $\bar \omega ^\intercal \mathbbm 1 = 1$. Consequently, we have
\begin{equation}
\|\bar v^{(l)}_{k|k} - \mathbbm 1 \bar v^*_{k|k}\|_{V} \leq \|{\Omega - \mathbbm 1 \bar{\omega}^\intercal} \|_{\textsc{op}}^l \|\bar v^{(0)}_{k|k} - \mathbbm 1 \bar v^*_{k|k}\|_{V}
\end{equation}
which implies (\ref{eq:convergence}). 
\end{proof}

Lastly, we remark that a finite number of average consensus steps can be used in practice (i.e., $\alpha < \infty)$, with $\|{\Omega - \mathbbm 1 \bar{\omega}^\intercal} \|_{\textsc{op}}$ dictating the rate of convergence. Using (\ref{eq:csd}), it can be seen that each weighted average consensus iteration lowers the Cauchy-Schwarz divergence between the sensors' posterior intensities (which is a measure of their information difference \cite{gostar2017cauchy}).

\subsubsection{Comparison with Existing Results}
Our analysis in this section extends the existing results on consensus (which has largely only considered consensus of vectors in Euclidean spaces)
to the case where consensus is sought on functions in $L^p$. For consensus in Euclidean spaces, a set of sufficient conditions for convergence is the following: 
\begin{itemize}
\item $\Omega_{ij} \geq 0$ whenever $(j,i)\in \mathcal E$,
\item $\Omega_{ii}>0$ ${\forall i\in\mathcal V}$, and 
\item the graph $\mathcal G$ is strongly connected
\end{itemize}
in which case, the fact that $\|{\Omega - \mathbbm 1 \bar{\omega}^\intercal} \|_{\textsc{op}}<1$ follows from the Perron-Frobenius theorem \cite[Lemma 1]{corless2012consensus}. Thus, the foregoing sufficient conditions for convergence are stronger than the ones we obtained in Section \ref{sec:subsec_error_analysis}.
Using the protocol in \cite{corless2012consensus}, the weights $\Omega_{ij}$ may be chosen in a distributed manner, i.e., each sensor is able to choose the weights in real-time without needing global knowledge of the topology of $\mathcal G$.

\subsection{Limitations of the Average Consensus Approach}
Although the average consensus algorithm (\ref{eq:average_consensus}) guarantees asymptotic convergence to the WAA, it has several limitations. Firstly, it requires each sensor to transmit $J_{i,k|k}$ GM components to its out-neighbors. This limits the number of targets that can be tracked simultaneously by the sensor network when the communication channels between the sensors have limited bandwidth, as is usually the case in wireless sensor networks. Secondly, the large number of GM components accumulated at each sensor increase the computational complexity of the subsequent local PHD filtering steps. Lastly, as noted in \cite{li2018partial}, the consensus step can exacerbate the problem of false positives in WAA fusion; small GM components corresponding to the false positive detections (i.e., clutter measurements which are misidentified as targets) can propagate through the cycles (closed loops) of the graph $\mathcal G$, leading to an overall feedback effect. In this way, false positives can get amplified during the fusion step, reducing the estimation performance of the sensor network. We address each of these concerns in the next section, by proposing a distributed protocol for dissemination and fusion of GM components that ensures asymptotic consensus and suppression of false positives, while having the same (or lower) communication bandwidth requirement than the existing distributed PHD filtering algorithms.

\section{Multi-Sensor Fusion under Communication Constraints}
\label{sec:comm_constr}

In order to carry out the fusion step of (\ref{eq:average_consensus}), the sensors must transmit a certain number of GM components to their out-neighbors. Since the dimension of the state-space $\mathcal X$ is fixed as $d_x$, each GM component is specified by $d_x + \tfrac{1}{2}d_x(d_x + 1) + 1$ floating point numbers, where the individual terms correspond to the mean vector, the covariance (which is a symmetric matrix), and the scalar weight of the GM component, respectively. In this section, we consider the case where the communication bandwidth is limited, such that the maximum number of GM components that may be transmitted across the inter-sensor communication channels at a given timestep is denoted as $B$, with $B \geq 1$.
In order to limit the communication complexity of the distributed GM-PHD algorithm, the authors of \cite{li2018partial} and \cite{li2020parallel} proposed the following heuristics for selecting the $J_{i,k|k}$ GM components that are to be transmitted at each sensor:
\begin{itemize}
    \item \textit{Rank Rule}: Rank the GM components based on their weights (highest to lowest) and select the top $B$ components
    \item \textit{Threshold Rule}: Fix a threshold, and transmit any GM components that have weights greater than this threshold
\end{itemize}
In either case, the GM components with small weights are not transmitted. This has the advantage of suppressing false positives, which typically correspond to the components having small weights. However, convergence to the WAA (\ref{eq:WAA_def}) cannot be guaranteed if the rank or threshold rule is used, as the GM components with small weights are never communicated between sensors.
Consequently, this approach is referred to as \textit{partial consensus} by its authors.

In order to guarantee asymptotic convergence of the average consensus algorithm to the WAA, we propose a new rule for selecting the GM components, which we call the \textit{sampling rule}. For simplicity, we assume that the communication is based on wireless broadcast, i.e., each sensor transmits the same message to all of its out-neighbors. Let $\mathcal B_{i}$ be a sequence of indices (numbers between $1$ and $J_{i,k|k}$) corresponding to the GM components transmitted by sensor $i$ to its out-neighbors at timestep $k$, where we omit the timestep $k$ for brevity. The total number of distinct indices in $\mathcal B_i$ is at most $B$, as per the communication constraint. The indices in $\mathcal B_i$ are chosen through random sampling, where the sampling is carried out either with or without replacement.

\subsection{Sampling with Replacement}
In order to construct $\mathcal B_i$, sensor $i$ generates a sequence of random samples of indices with replacement. In each sample,
the index $l$ is chosen with probability $p_{i}^{(l)}$, with 
\begin{equation}
    \sum_{l=1}^{J_{i,k|k}} p_{i}^{(l)} = 1
\end{equation}
Thus, the elements of $\mathcal B_i$ are $i.i.d.$ samples of a categorical random variable and can be thought of as the outcomes of rolling a $J_{i,k|k}$-sided biased die $|\mathcal B_i|$ times. Due to the limited communication bandwidth, the sampling process is terminated when the number of distinct indices in $\mathcal B_i$ equals $B$.
Let $\zeta^{(l)}_{i}$ denote the number of times that the index $l$ appears in $\mathcal B_i$. Observe that $(\zeta^{(1)}_{i}, \zeta^{(2)}_{i}, \dots, \zeta^{(|J_{i,k|k}|)}_i)$ follow the multinomial distribution, with $\mathbb E[\zeta^{(l)}_{i}]=p_{i}^{(l)} |\mathcal B_i|$.

Once $\mathcal B_i$ is generated, the sensor transmits the corresponding GM components to its out-neighbors. If there are any repeated components, the sensor transmits each of these components (in addition to the number of times they were sampled) only once, in order to avoid redundancy in the communication. Thus, each sensor transmits a random function (i.e., a random GM) to its out-neighbors at each timestep, given by
\begin{align}  &\sum_{l=1}^{J_{i,k|k}}\zeta^{(l)}_{i} \hspace{2pt} \tilde w^{(l)}_{i,k|k}\mathcal N(m^{(l)}_{i,k|k}, P^{(l)}_{i,k|k})
    \label{eq:random_GM}
\end{align}
where $\mathcal N(m,P)$ denotes a GM component having the mean vector $m\in \mathbb R^{d_x}$ and covariance $P\in \mathbb R^{d_x \times d_x}$.
Note that the weights of the transmitted GM components $\tilde w^{(l)}_{i,k|k}$ are different from the weights used in the local GM-PHD filtering, $w^{(l)}_{i,k|k}$.
To ensure that the fixed point of the average consensus algorithm is preserved (in expectation) by the sampling rule, we require that the expected value of the random GM is equal to the posterior intensity of sensor $i$, i.e.,
\begin{align}
\mathbb E\left[\sum_{l=1}^{J_{i,k|k}}\zeta^{(l)}_{i} \hspace{2pt} \tilde w^{(l)}_{i,k|k}\mathcal N(m^{(l)}_{i,k|k}, P^{(l)}_{i,k|k})\right] = v_{i,k|k}
\label{eq:exp_randomGM}
\end{align}
Equivalently,
we have
\begin{align}
\mathbb E[\zeta^{(l)}_{i}]\hspace{2pt}\tilde w^{(l)}_{i,k|k} = w^{(l)}_{i,k|k} \label{eq:exp_weight_condition}\\ \hspace{1pt}  \tilde w^{(l)}_{i,k|k} = \frac{w^{(l)}_{i,k|k}}{p_{i}^{(l)}|\mathcal B_i|}
\label{eq:weight_condition}
\end{align}
Equation (\ref{eq:weight_condition}) imposes a constraint on the sampling probabilities $p_{i}^{(l)}$ and the weights $\tilde w_{i,k|k}^{(l)}$, but does not uniquely specify them. 
We propose the following choice of sampling probabilities:
\begin{equation}
p_{i}^{(l)}=\frac{w^{(l)}_{i,k|k}}{\sum_{l=1}^{J_{i,k|k}} w^{(l)}_{i,k|k}}
\label{eq:pl_choice}
\end{equation}
for which the corresponding weights are given by (\ref{eq:weight_condition}), as
\begin{equation}
\tilde w^{(l)}_{i,k|k} = \frac{\sum_{l=1}^{J_{i,k|k}} w^{(l)}_{i,k|k}}{|\mathcal B_i|} \triangleq \tilde w_{i,k|k}
\label{eq:final_w_tilde_def}
\end{equation}
The proposed choice of sampling probabilities, given by (\ref{eq:pl_choice}) and (\ref{eq:final_w_tilde_def}), has the clear advantage that the weights $\tilde w^{(l)}_{i,k|k}$ no longer need to be transmitted for each component, a single number $\tilde w_{i,k|k}$ may be transmitted instead, further reducing the communication requirement of the algorithm. Additionally, observe that if the sampling probabilities $p_{i}^{(l)}$ are chosen as per (\ref{eq:pl_choice}), then GM components with higher weights are transmitted more often than those with smaller weights,
so that the components with small weights only survive the fusion step if the targets corresponding to them are persistently detected. In this way, false positive detections that do not survive successive local PHD filtering steps are suppressed by the inter-sensor fusion step. The ability of the algorithm to suppress false alarms as well as the reduced communication cost motivate the choice of $p_{i}^{(l)}$ as (\ref{eq:pl_choice}).

\begin{algorithm}
\caption{Random Sampling Rule (with Replacement)}
\begin{algorithmic}[1]
\vspace{2pt}
\REQUIRE The maximum allowable number of GM components, $B\geq 1$ \\
\vspace{3pt} 
\WHILE{$\mathcal B_i$ has at most $B$ distinct indices}
\STATE Randomly sample an index from the sequence $\left(1, 2, \dots, J_{i,k|k}\right)$ with the corresponding sampling probabilities as given by (\ref{eq:pl_choice}).
\ENDWHILE
\STATE Each sensor broadcasts the following GM to its out-neighbors:
\[\tilde w_{i,k|k} \sum_{l=1}^{J_{i,k|k}}\zeta^{(l)}_{i} \hspace{2pt} \mathcal N(m^{(l)}_{i,k|k}, P^{(l)}_{i,k|k})
\]
where $\zeta^{(l)}_i$ is the number of times $l$ occurs in $\mathcal B_i$.
\end{algorithmic}
\label{alg:sampling}
\end{algorithm}
Thus, random sampling (as described in Algorithm \ref{alg:sampling}) can be used to replace step $3$ of Algorithm \ref{alg:con}, thereby limiting the communication cost of the distributed GM-PHD filter. From step $4$ of Algorithm \ref{alg:sampling}, we see that at each sensor, the sampling rule requires the  communication of $B$ GM components, $1$ floating point number (which is $\tilde w_{i,k|k}$), and less than $B$ integers (i.e., the numbers $\zeta_i^{(l)}$) at each iteration of the average consensus step.

Observe that, even though each sensor communicates a random GM to its neighbors, the integral of the GM is a deterministic quantity:
\begin{align}   
\tilde w_{i,k|k} \sum_{l=1}^{J_{i,k|k}}\zeta^{(l)}_{i} = \tilde w_{i,k|k} |\mathcal B_i| = \sum_{l=1}^{J_{i,k|k}} w^{(l)}_{i,k|k}
\end{align}
where in the second equality, we used (\ref{eq:final_w_tilde_def}).
Thus, while the posterior intensities of the sensors asymptotically converge to the WAA in expectation, the integrals of the intensities (i.e., the cardinality estimates of the sensors) converge deterministically.

\subsection{Sampling without Replacement}
An alternative strategy for constructing $\mathcal B_i$ is to sample the indices without replacement. As each index occurs in $\mathcal B_i$ at most once, $\mathcal B_i$ can be interpreted as a set. The random variable $\zeta^{(l)}_{i}$ defined in the previous section becomes the indicator variable $\mathbbm 1_{\lbrace l\in \mathcal B_i \rbrace}$, where
\begin{align}
    \mathbbm 1_{\lbrace l\in \mathcal B_i \rbrace} = \begin{cases}
    \begin{array}{cl}
         1& \text{if label}\ l\ \text{is in }\mathcal B_i,  \\
         0& \text{otherwise.}
    \end{array}
    \end{cases}
\end{align}
so that the condition (\ref{eq:exp_weight_condition}) now becomes
\begin{equation}
    \mathbb E[\mathbbm 1_{\lbrace l\in \mathcal B_i \rbrace}]\hspace{2pt}\tilde w^{(l)}_{i,k|k} = P (l\in \mathcal B_i )\hspace{2pt}\tilde w^{(l)}_{i,k|k} = w^{(l)}_{i,k|k} 
\end{equation}
where $P({}\cdot{})$ denotes the probability of an event. Thus, once the 
probabilities $P (l\in \mathcal B_i )$ are computed,
the weights $\tilde w^{(l)}_{i,k|k}$ must be chosen accordingly, as
\begin{equation}
    \tilde w^{(l)}_{i,k|k} = \frac{w^{(l)}_{i,k|k}}{P( l\in \mathcal B_i )}
    \label{eq:no_replacement_weights}
\end{equation}
An efficient method for sampling without replacement can be found in \cite{Efraimidis2015}, which is able to assign a higher probability to components with higher weights, thereby incorporating the ability to suppress false positives into the fusion step.

Note that it is no longer straightforward to specify the inclusion probabilities of the indices $P( l\in \mathcal B_i )$ \textit{a priori}, like we did in (\ref{eq:pl_choice}). To see this, consider the case where $B=J_{i,k|k}$, which fixes the probabilities as $P( l\in \mathcal B_i )=1$ for all the indices. Thus, if we were to specify the probabilities $P( l\in \mathcal B_i )$ \textit{a priori}, it can make the sampling problem infeasible; the same is true even when $B<J_{i,k|k}$ \cite[Example 2]{Efraimidis2015}. Moreover, when using the sampling without replacement rule, in addition to transmitting the mean vectors $m^{(l)}_{i,k|k}$ and covariances $P^{(l)}_{i,k|k}$ of the selected components, the sensors must transmit the modified weights $\tilde w_{i,k|k}$ as well. In contrast, the proposed sampling with replacement rule (given in Algorithm \ref{alg:sampling}) allows one to specify the sampling probabilities \textit{a priori}, and also has a lower communication bandwidth requirement.

In summary, the proposed distributed GM-PHD algorithm uses a local GM-PHD filter at each sensor to update its multi-target estimate, followed by a given number of average consensus steps (as per Algorithm \ref{alg:con}) to fuse the multi-target estimates of the sensors. Additionally, the proposed random sampling rule (given in Algorithm \ref{alg:sampling}) is used to limit the communication bandwidth of the approach, while ensuring its asymptotic optimality.

\section{Simulations}
\label{sec:simulation}
In this section, the performance of the proposed distributed GM-PHD filtering method is evaluated against three different algorithms from the literature:
\begin{itemize}
\item GM-PHD filter without consensus \cite{vo2006gmphd}, 
\item GM-PHD filter with full consensus (in which all the GM components are communicated between neighbors) \cite{battistelli2020differentFOVs}, and 
\item GM-PHD filter with partial consensus (which uses the rank rule to limit the number of GM components that are communicated) \cite{li2018partial}
\end{itemize}

No inter-sensor communication is carried out in the GM-PHD filter without consensus. In the other three algorithms, each sensor implements a total of $\alpha \geq 1$ average consensus steps between successive local PHD filtering steps (Algorithm \ref{alg:con}). The partial consensus algorithm uses the rank rule (described in Section \ref{sec:comm_constr}) to limit the number of GM components which are broadcast by each sensor to be $B$, where we set $B=5$.
Similarly, the proposed algorithm uses the sampling with replacement rule (Algorithm \ref{alg:sampling}) to limit the number of distinct GM components that are transmitted by a given sensor to $B$. Consequently, the communication bandwidth requirement of the proposed algorithm is at most that of the GM-PHD filter with partial consensus. The communication costs of the different algorithms can be ordered as follows:
\begin{equation*}
\text{Proposed Method} \leq \text{Partial Consensus} \ll \text{Full Consensus}
\end{equation*}
\noindent
While the full consensus algorithm may be unsuitable for practical applications due to its high communication cost, it serves as a benchmark for evaluating the multi-target tracking capability of the proposed distributed GM-PHD algorithm.

\subsection{Simulation Scenario}
We use a multi-target tracking simulation scenario to evaluate the distributed GM-PHD algorithms. The simulation scenario described herein is similar to the one considered in \cite{vo2006gmphd}, with the key difference being that we consider a greater (time-varying) number of targets. Moreover, Ref. \cite{vo2006gmphd} considers the case of single-sensor multi-target tracking, whereas we consider the multi-sensor multi-target tracking case.

Consider a multi-target set having a 4-dimensional state space, such that the state of each target is composed of its $2$ dimensional position and velocity vectors. Thus, the $l^{th}$ target's state is \[x_k^{(l)}=\begin{bmatrix}
\textrm{pos}_{k}^\intercal & \textrm{vel}_{k}^\intercal\end{bmatrix}^{\intercal}\] where $\textrm{pos}_{k} \in \mathbb R^2$ and $\textrm{vel}_{k}\in \mathbb R^2$ denote the target's position and velocity vectors at timestep $k$, respectively.
%
The targets follow the linear Gauss-Markov dynamical model (\ref{eq:dynamic_model}), with
\begin{equation}
    F_{k} = \begin{bmatrix}
        I & h I \\
        0 & I
    \end{bmatrix}
\ \ \text{and}\ \ 
    Q_{k} = 9 \begin{bmatrix}
        \frac{h^4}{4} I & \frac{h^3}{2} I \vspace{2pt}\\
         \frac{h^3}{2} I & h^2 I 
    \end{bmatrix} m^2s^{-4}, 
\end{equation}
where $h=1s$ is the sampling time. The target birth intensity $\gamma_k$ is chosen as
\begin{equation}
    \gamma _k =  0.2 \sum_{l=1}^{10}\mathcal N\left(m_{\gamma}^{(l)}, \textrm{diag}\big(\begin{bmatrix} 100 & 100 & 25 & 25\end{bmatrix}^\intercal\big)\right)
\end{equation}
where given a vector $v$, $\textrm{diag}(v)$ denotes the diagonal matrix whose diagonal elements are the components of $v$.
The mean vectors $m_{\gamma, k}^{(l)}$ represent the initial guesses of the position and velocity vectors of the targets. The mean vectors of the initial positions are as depicted in Figure \ref{fig:trajectories} by the `$\triangle$' symbols, whereas the mean vectors of the initial velocities are set to $[0ms^{-1}\ 0ms^{-1}]^\intercal$.
Given that a target has the state $\zeta$ at a given timestep, the intensity corresponding to its spawned targets is given by
\begin{equation}
    \sigma_{k|k-1}(\cdot \vert \zeta) =  0.1\hspace{3pt}\mathcal N\Big(\zeta, \text{diag}\big(\begin{bmatrix} 100 & 100 & 400 & 400 \end{bmatrix}^\intercal\big)\Big)
\end{equation}
The survival probability $p_{k}^S(x)$ of each target is chosen according to the model described in \cite{sun2023gaussian}, such that $p_{k}^S(x)$ is close to $1$ when the target is inside the surveillance region, and close to $0$ otherwise. 

The trajectories of the targets are shown in Fig.\ref{fig:trajectories}, and the start and end times of their trajectories are given in Table \ref{table_start_end_trajectories}. Figure \ref{truth_cardinality} shows the number of targets at each timestep of the simulation, which varies between $6$ to $9$.

\begin{figure}[h]
\centering    
\includegraphics[width=0.39\textwidth]{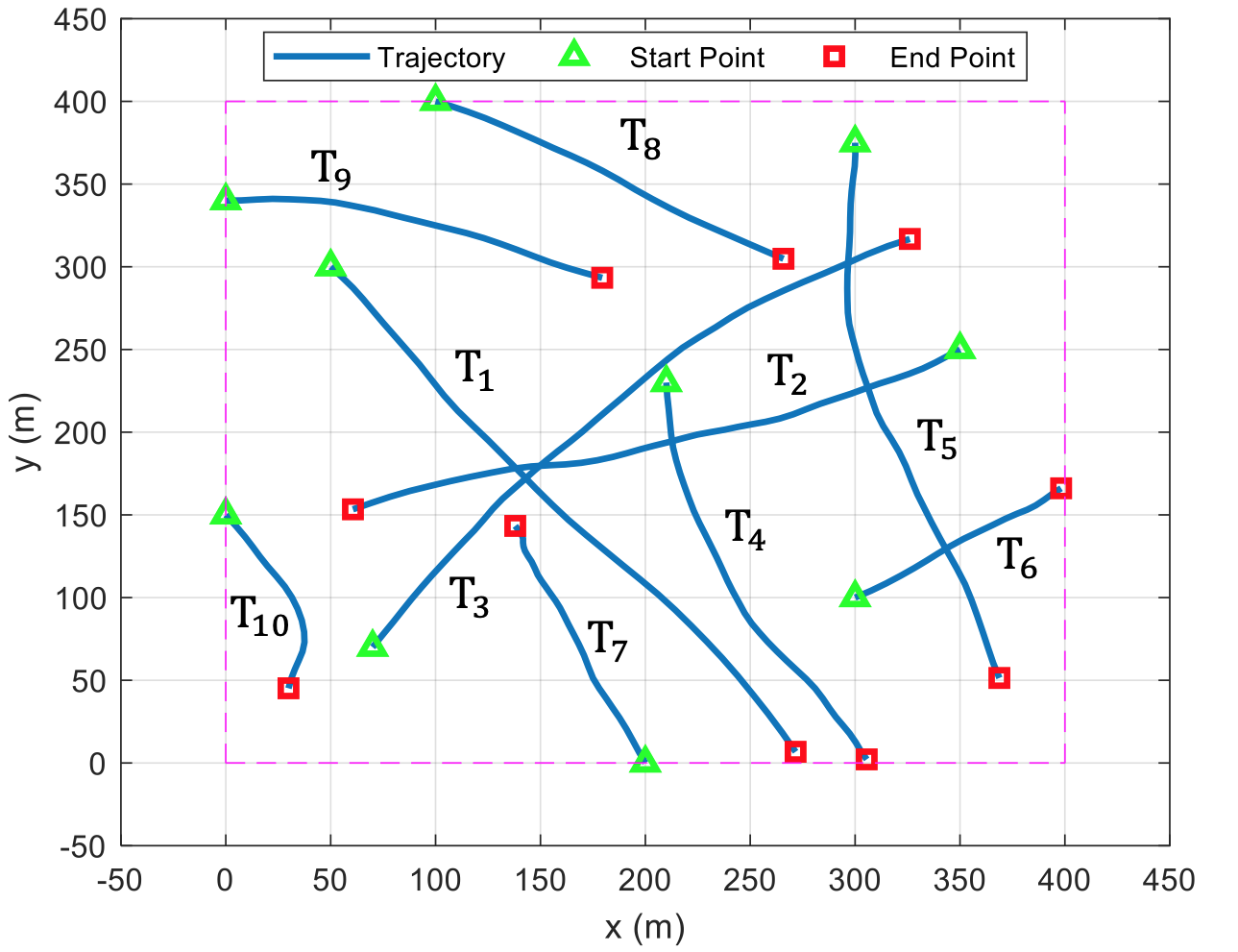}    
    \caption{The true trajectories of the $10$ targets (which are labeled as $T_1, T_2, \dots, T_{10}$) are shown by the solid lines. The start and end points of each trajectory are denoted by the symbols `$\triangle$' and `$\square$', respectively. The dashed line represents the boundary of the surveillance region.} 
    \label{fig:trajectories}
\end{figure}
\begin{figure}[h]
\centering
\includegraphics[width=0.5\textwidth]{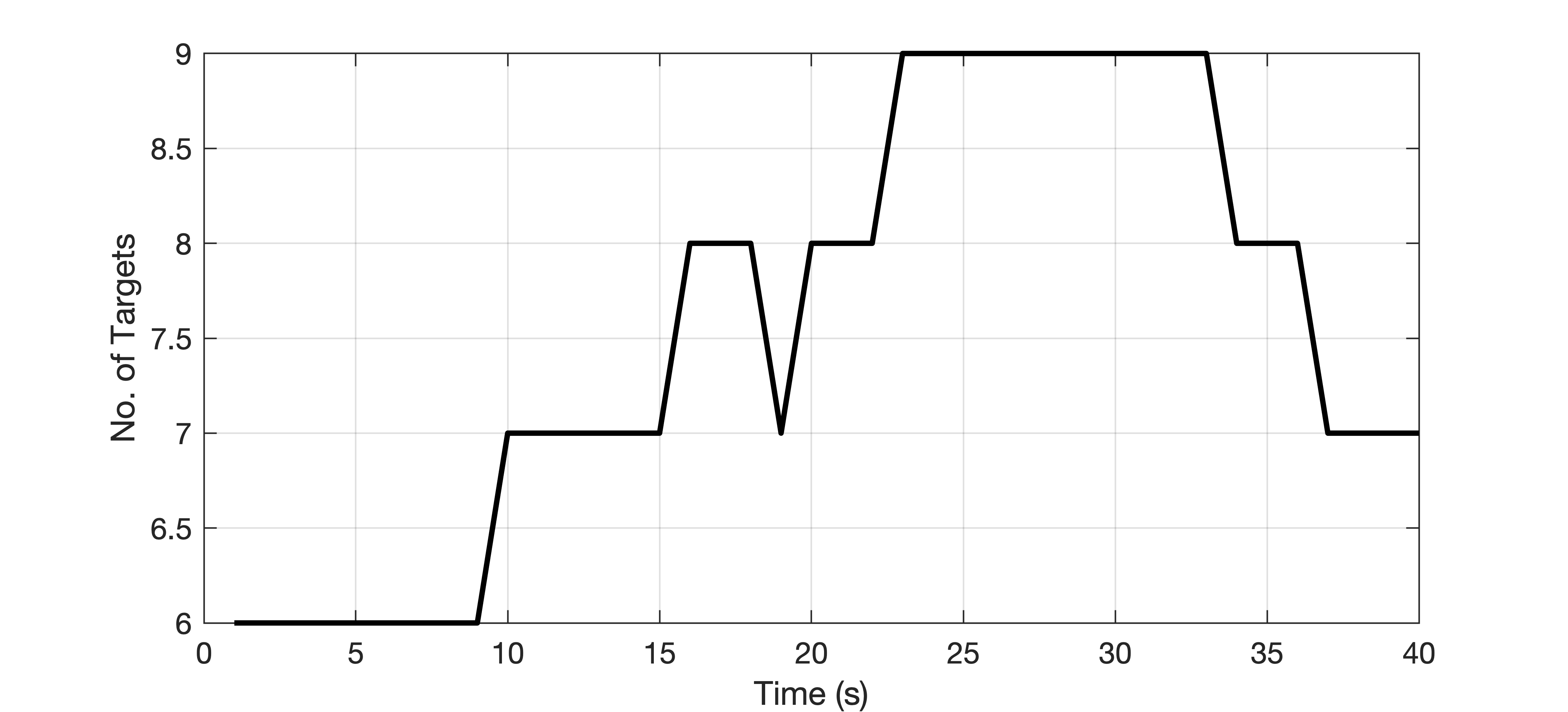}    
    \caption{Total number of targets (i.e., $|X_k|$) over time.} 
    \label{truth_cardinality}
\end{figure}

\begin{table}[h]
\centering
\caption{Starting and ending timesteps of the target trajectories}
\label{table_start_end_trajectories}
    \begin{tabular}{|c|c|c|c|c|c|c|} 
        \cline{1-3} \cline{5-7}
        Target & Start ($s$) & End ($s$) && Target  & Start ($s$) & End ($s$) \\ 
        \cline{1-3} \cline{5-7}
        $T_1$ & 1 & 34 && $T_6$ & 1 & 19 \\ 
        \cline{1-3} \cline{5-7}
        $T_2$ & 1 & 40 && $T_7$ & 10 & 40 \\ 
        \cline{1-3} \cline{5-7}
        $T_3$ & 1 & 40 && $T_8$ & 20 & 40 \\ 
        \cline{1-3} \cline{5-7}
        $T_4$ & 1 & 37 && $T_9$ & 16 & 40 \\ 
        \cline{1-3} \cline{5-7}
        $T_5$ & 1 & 40 && $T_{10}$ & 23 & 40 \\ 
        \cline{1-3} \cline{5-7}
    \end{tabular}
\end{table}


The targets are observed using a sensor network having $6$ sensors, over a  $400m\times400m$ surveillance region. The sensor network topology is assumed to be bidirectional, as shown in Figure \ref{fig_network}. 
\begin{figure}[t]
\centering
    \includegraphics[width=0.2\textwidth]{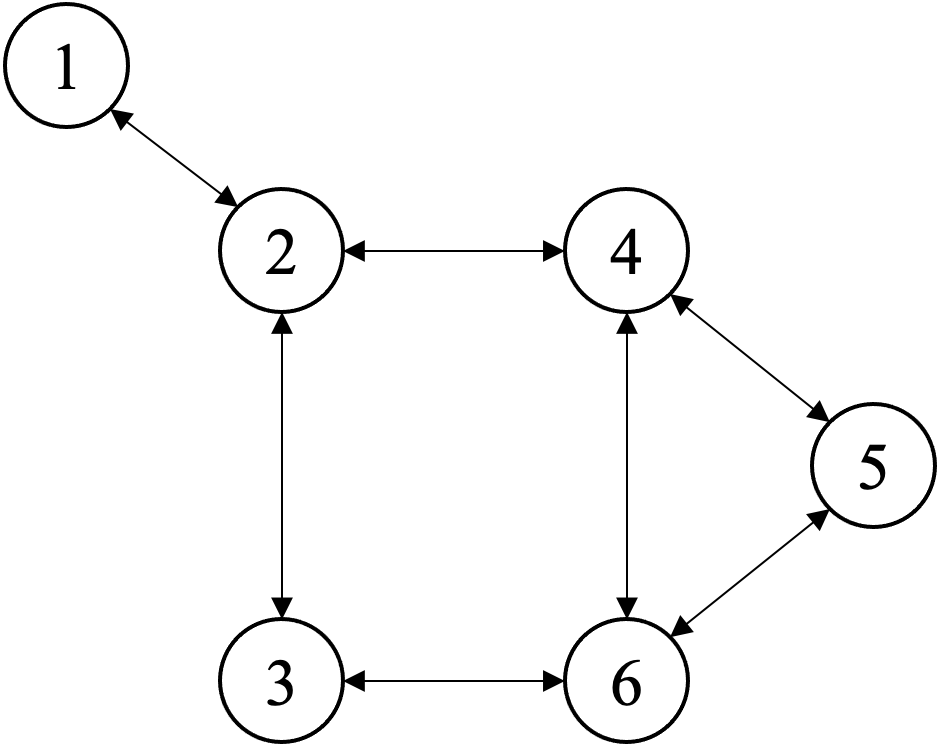}
    \caption{Directed graph of the sensor network with 6 sensors.} 
    \label{fig_network}
\end{figure}
The consensus weights, $\Omega_{ij}$, are chosen as per the \textit{Metropolis} rule described in \cite{xiao2006distributed} (which is commonly used in the literature on distributed PHD filtering \cite{li2018cardinality, li2020parallel}), giving us
\begin{equation}
    \Omega = \begin{bmatrix}
        0.8 & 0.2 & 0 & 0 & 0 & 0\\
        0.2 & 0.4 & 0.2 & 0.2 & 0 & 0\\
        0 & 0.2 & 0.6 & 0 & 0 & 0.2\\
        0 & 0.2 & 0 & 0.4 & 0.2 & 0.2\\
        0 & 0 & 0 & 0.2 & 0.6 & 0.2\\
        0 & 0 & 0.2 & 0.2 & 0.2 & 0.4
    \end{bmatrix}.
\end{equation}

Similar to the target survival probability, the detection probability $p_{i,k}^D(x)$ is set to $0.98$ within the surveillance region and $0$ outside it. The measurement is given by (\ref{eq:measurement_model}), where
\begin{equation}
    H_{k} = \begin{bmatrix}
        1 & 0 & 0 & 0\\
        0 & 1 & 0 & 0
    \end{bmatrix}\ \ \text{and}\ \ 
    R_{k} = \begin{bmatrix}
        25 & 0\\
        0 & 25
    \end{bmatrix}m^2.
\end{equation}
The clutter is modeled as a Poisson RFS having the intensity
\begin{equation}
    \kappa_{k}(z) = \lambda_{c} A u(z)
\end{equation}
where $u(\cdot)$ is the uniform distribution over the surveillance region (whose integral is equal to $1$), $A=1.6 \times 10^5 m^2$ is the area of the surveillance region, and $\lambda_{c}$ is the average number of clutter measurements per unit volume, with $\lambda_{c}=3.125\times10^{-5}m^{-2}$. As the integral of $\kappa_{k}(z)$ over the surveillance region is $5$, an average of $5$ clutter measurements are generated
at each sensor at each timestep.

In addition to the simulation parameters given above, the local GM-PHD filter has additional parameters related to the pruning of GMs at each timestep, whose description can be found in \cite[Table II]{vo2006gmphd}. The pruning threshold is set to $10^{-5}$, the merging threshold is set to $15$, and the maximum number of GM components is set to $50$. At each sensor, the GM components whose weights are greater than or equal to $0.5$ are identified as the targets. In the literature, \textit{target extraction} refers to the process by which, at each sensor, the GM components whose weights are higher than a given threshold are identified as the targets. The threshold for the extraction of targets is set to $0.5$.

\subsection{Results}
The Optimal Sub-Pattern Assignment (OSPA) distance \cite{schuhmacher2008consistent} is used to quantitatively evaluate the multi-target tracking performance of the proposed algorithm at each sensor, with the order parameter of OSPA chosen as $1$ and the cut-off parameter chosen as $100$. Generally, a high OSPA value is indicative of one or more of the following: (i) a large amount of discrepancy between the extracted target states and the true states of the targets, (ii) under-estimation of the number of targets (e.g., due to missed detections), or (iii) over-estimation of the targets (e.g., due to false positives).
We define the \textit{network OSPA} distance as the average of the OSPA distances attained at all the sensors in the network at each timestep. Furthermore, the \textit{time-averaged network OSPA} distance is the average of the network OSPA distances across all the timesteps of the simulation.
In the following simulation results, each of the metrics is averaged over 100 Monte Carlo simulations; the number of targets, their initial positions, and their starting/ending timesteps (as given in Table \ref{table_start_end_trajectories}) are kept fixed across the simulations, whereas the process and measurement noise as well as the clutter measurements are randomly sampled in each simulation.

Figure \ref{fig_timeAvg_OSPA_B5_t0_t6} shows the time-averaged network OSPA distance for each algorithm, as a function of the number of inter-sensor fusion steps, $\alpha$, which is varied between $0$ and $6$. It can be observed that each method tends to achieve better multi-target tracking performance as $\alpha$ is increased. In particular, the proposed method achieves a multi-target tracking performance that is close to that of the full consensus method, while having a communication cost that is at most that of the partial consensus method borrowed from the literature. 

\begin{figure}[h]
\centering
    \includegraphics[width=0.51\textwidth]{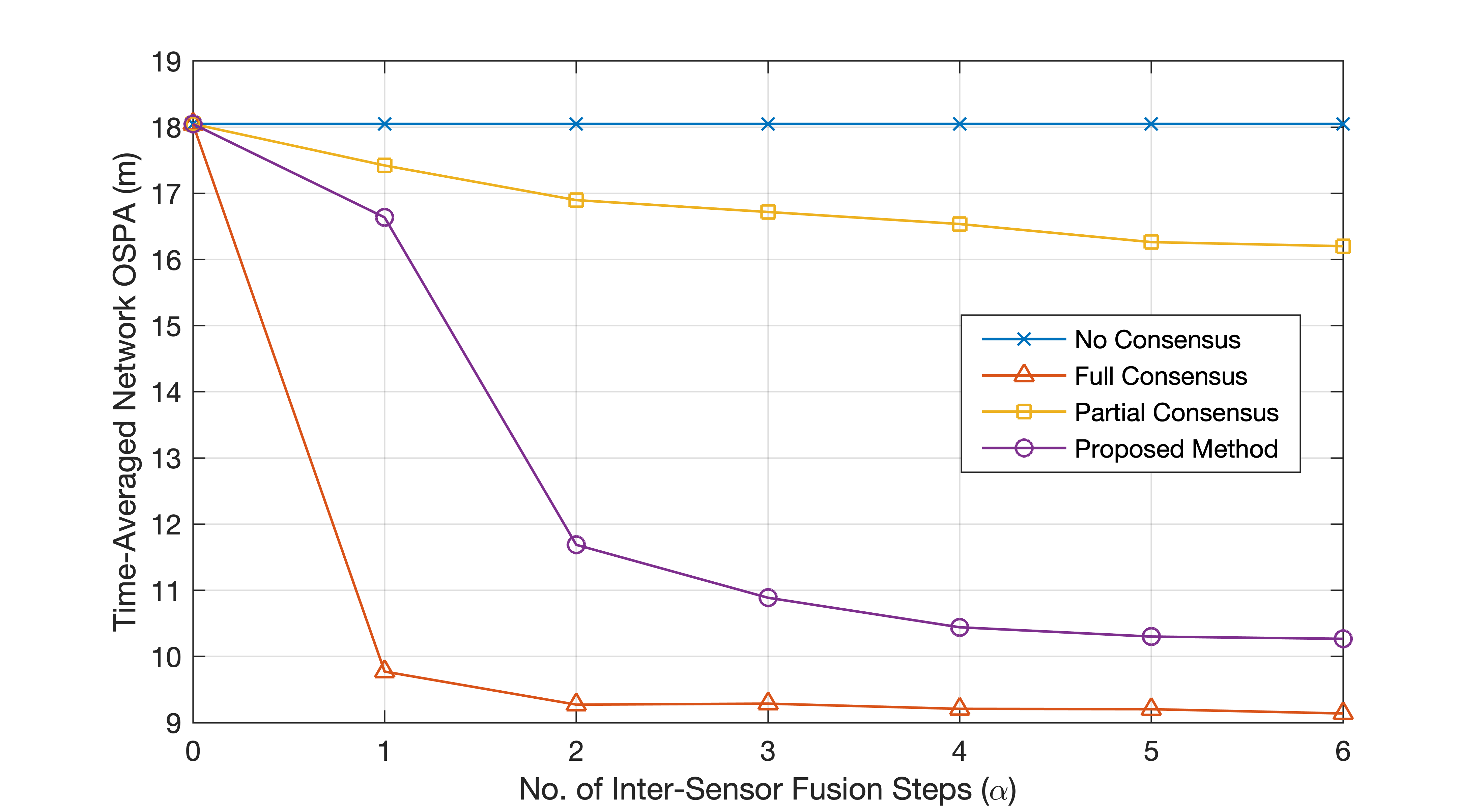}
    \caption{Time-averaged network OSPA distances plotted as a function of $\alpha$, which represents the number of inter-sensor fusion steps used.}
    \label{fig_timeAvg_OSPA_B5_t0_t6}
\end{figure}

The individual OSPA distances of sensor 6 are plotted in Figure \ref{fig_localOSPA_sensor6} and the network OSPA distances are plotted in Figure \ref{fig_networkOSPA}. Figures \ref{fig_localOSPA_sensor6} and \ref{fig_networkOSPA} show that the OSPA distances of sensor 6 are close to the corresponding network average values. 
In each algorithm, the OSPA distances are higher in the timesteps following a change in the number of targets present in the surveillance region; this can be observed by comparing Figure \ref{fig_localOSPA_sensor6} with Figure \ref{truth_cardinality}. At the other timesteps, the OSPA distances of the proposed algorithm are considerably lower than those of the partial consensus method, irrespective of the number of inter-sensor fusion steps used. As discussed in Section \ref{sec:comm_constr}, the improved performance of the proposed distributed GM-PHD algorithm can be attributed to the fact that, when using the random sampling rule, each GM component of a given sensor has a non-zero probability of being transmitted to its neighbors. In contrast, a sensor using the partial consensus method only transmits the top $B$ GM components (in terms of their weights) to its neighbors, where $B=5$ in this simulation scenario.

\begin{figure}[h]
    \centering
    \begin{subfigure}[b]{0.239\textwidth}
    \centering
    \includegraphics[width=\textwidth]{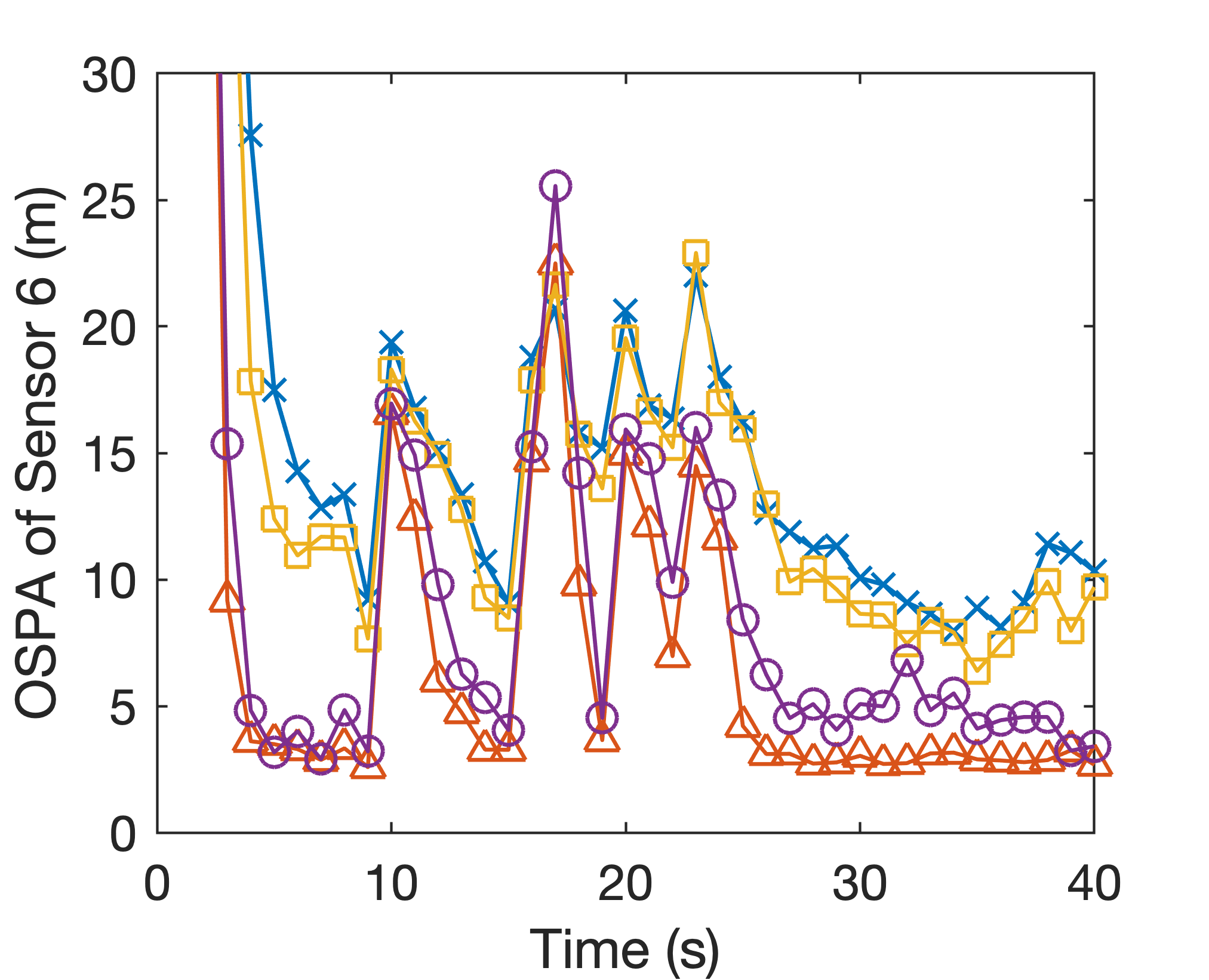}
    \caption{$\alpha=1$}
    \label{fig_localOSPA_B5_t1_sensor6}
    \end{subfigure}
    \begin{subfigure}[b]{0.241\textwidth}
    \centering
    \includegraphics[width=\textwidth]{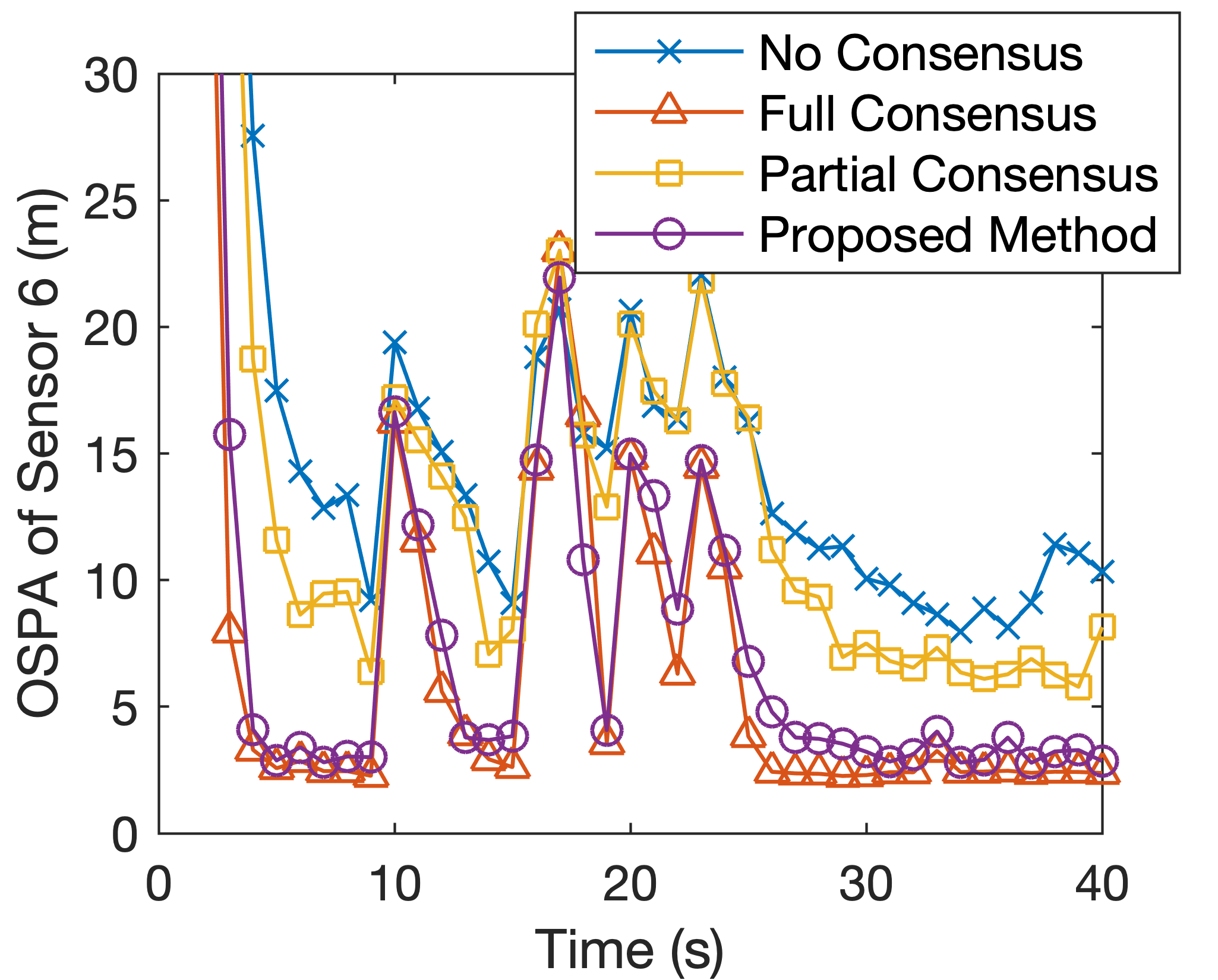}
    \caption{$\alpha=6$}
    \label{fig_localOSPA_B5_t6_sensor6}
    \end{subfigure}
    \caption{The OSPA distances attained at sensor 6 at each timestep of the simulation, where $\alpha$ represents the number of inter-sensor fusion steps used.}
    \label{fig_localOSPA_sensor6}
\end{figure}

\begin{figure}[h]
    \centering
    \begin{subfigure}[b]{0.237\textwidth}
    \centering
    \includegraphics[width=\textwidth]{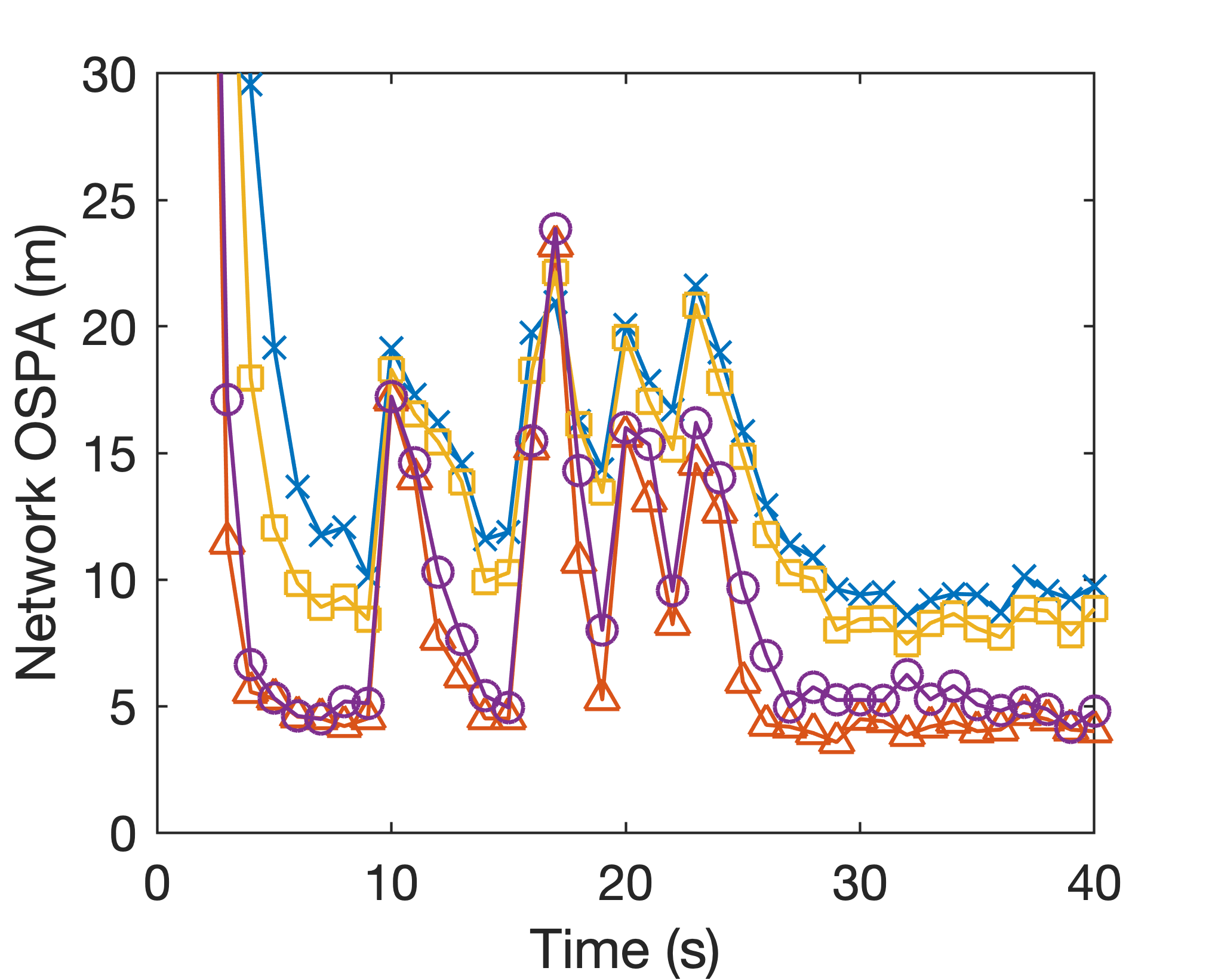}
    \caption{$\alpha=1$}
    \label{fig_networkOSPA_B5_t1}
    \end{subfigure}
    \begin{subfigure}[b]{0.243\textwidth}
    \centering
    \includegraphics[width=\textwidth]{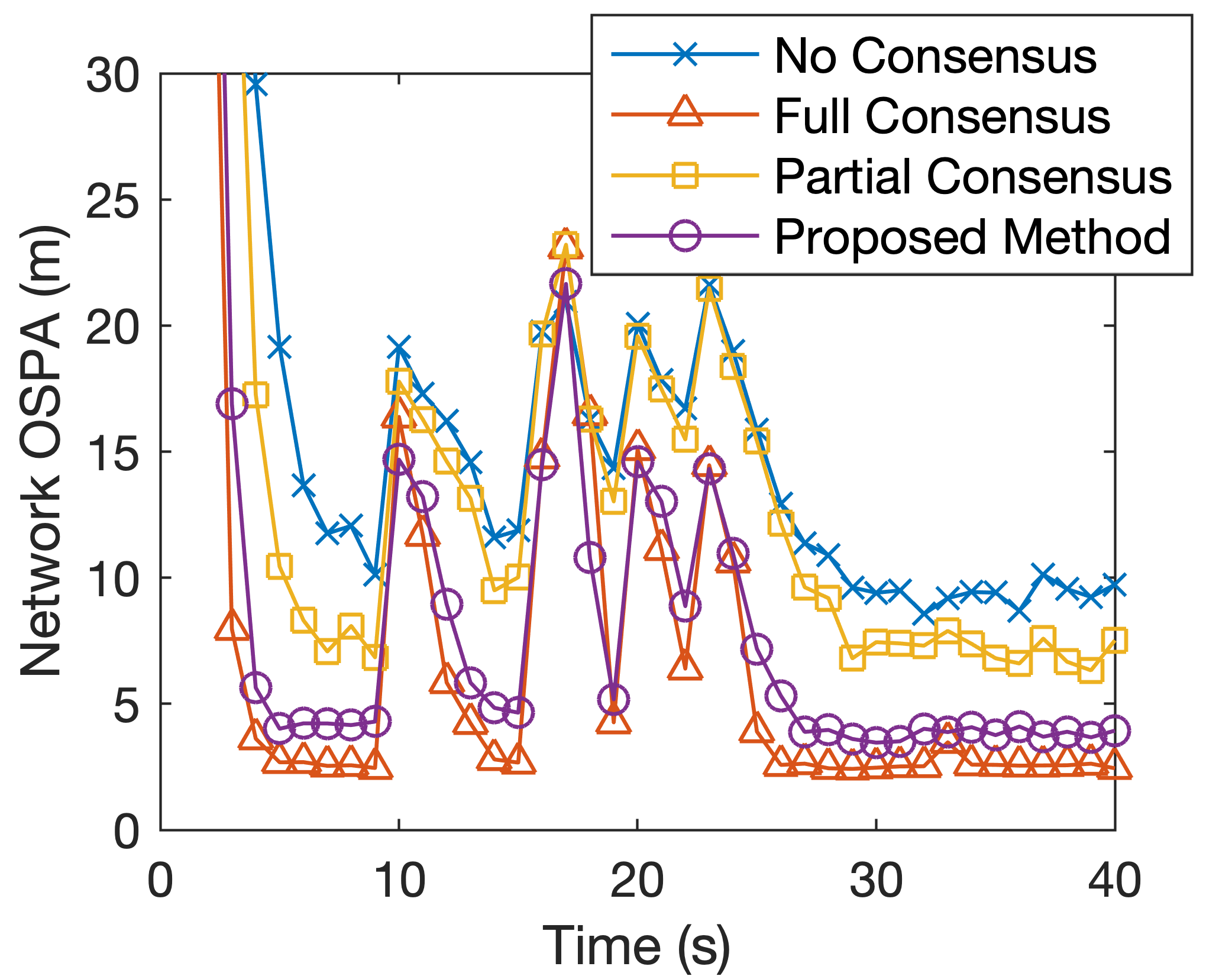}
    \caption{$\alpha=6$}
    \label{fig_networkOSPA_B5_t6}
    \end{subfigure}  
    \caption{The network OSPA distances attained at each timestep of the simulation, where $\alpha$ represents the number of inter-sensor fusion steps used.}
    \label{fig_networkOSPA}
\end{figure}

The computation time required for each filtering step of the algorithms considered in this section is shown in Figure \ref{fig_computeTime_B5}. It can be seen that, while the proposed method has a communication cost that is at most that of the partial consensus method, its computational cost is slightly higher than that of the latter. This can be attributed to the increased complexity of the random sampling rule as compared to the rank rule described in Section \ref{sec:comm_constr}. 
At the same time, the computational cost of the proposed algorithm is lower than that of the full consensus method, due to the lower number of GM components that need to be pruned and fused in the proposed method. Thus, the proposed distributed GM-PHD method strikes a desirable trade-off between computational cost, communication cost, and multi-target tracking performance.

\begin{figure}[h]
\centering
    \includegraphics[width=0.5\textwidth]{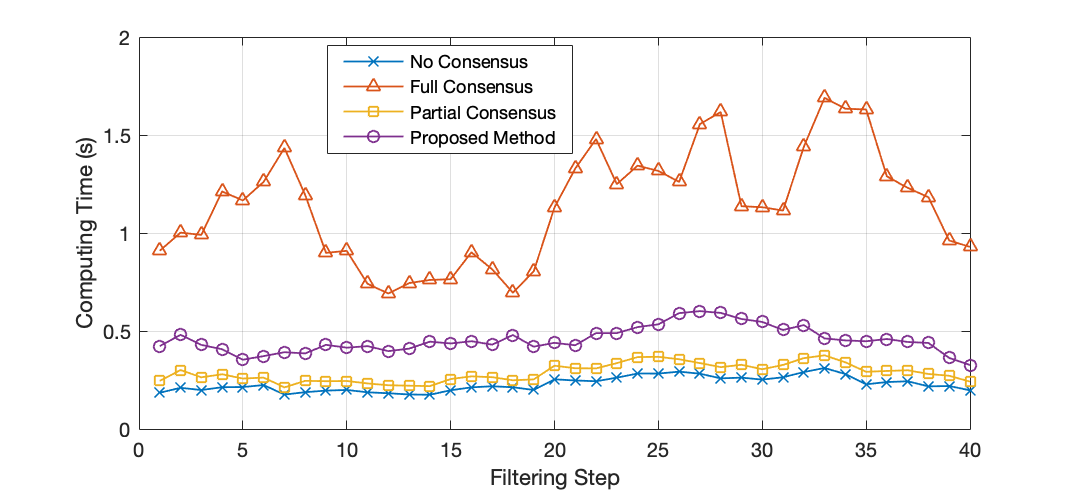}
    \caption{The computation time required for each filtering step of the algorithms considered in Section \ref{sec:simulation}.}
    \label{fig_computeTime_B5}
\end{figure}

\section{Conclusion}
\label{sec:conclusion}

In this paper, an asymptotically optimal distributed Gaussian-Mixture Probability Hypothesis Density (GM-PHD) filter was developed for multi-sensor multi-target tracking under communication constraints. Unlike the existing implementations of distributed GM-PHD filtering given in the literature, the proposed method is able to guarantee asymptotic convergence to an optimal multi-target state estimate, which corresponds to the weighted arithmetic average (WAA) of the posterior densities of the sensors. A novel result related to the consensus of functions in $L^p$ spaces was developed in order to establish the convergence properties of the consensus-based GM-PHD filtering method. Subsequently, a random sampling rule was developed, which uses probabilistic sampling of GM components to limit the amount of information that is transmitted between sensors in the inter-sensor fusion step, thereby limiting the communication cost of the proposed method. Through numerical simulations, it was shown that the proposed method (which uses the random sampling rule) exhibits a much better multi-target tracking performance than a comparable algorithm from the literature, while requiring the same (or lesser) amount of communication, thereby striking a desirable trade-off between the aforementioned metrics that is well-suited for distributed sensor networks having communication constraints.

\bibliographystyle{plain}
\bibliography{refs}

\end{document}